\documentclass[11pt]{article}
\usepackage{amsmath}
\usepackage{amsfonts}
\usepackage{amssymb}
\usepackage{amsthm}
\usepackage{xcolor}
\usepackage{natbib}
\usepackage[pagewise]{lineno}
\usepackage[affil-it]{authblk}
\newtheorem{propi}{Proposition}

\newtheorem{teo}{Theorem}
\newtheorem{lema}{Lemma}
\newtheorem{coro}{Corollary}

\newtheorem{defi}{Definition}
\newtheorem{exa}{Example}
\author[1,2]{Federico Fioravanti\footnote{f.fioravanti@uva.nl (corresponding author)}}
\author[4]{Iyad Rahwan}
\author[3]{Fernando Tohm\'e}
\affil[1]{Institute for Logic, Language and Information, University of Amsterdam, Amsterdam, The Netherlands}
\affil[2]{Instituto de Matem\'atica (INMABB), Departamento de Matem\'atica, Universidad Nacional del Sur (UNS)-CONICET, Bah\'ia Blanca, Argentina}
\affil[3]{Instituto de Matem\'atica (INMABB), Departamento de Econom\'ia, Universidad Nacional del Sur (UNS)-CONICET, Bah\'ia Blanca, Argentina}

\affil[4]{Center for Humans and Machines, Max Planck Institute for Human Development, Berlin, Germany}
\title{Classes of Aggregation Rules for Ethical Decision Making in Automated Systems\footnote{We are grateful to Ulle Endriss, Ariel Procaccia, Felix Brandt, and two anonymous reviewers for comments and suggestions that led to an improvement of the paper. This work began during Fioravanti's research stay at the MPI for Human Development. Thanks are due to the funding of the German Academic Exchange Service (DAAD). Fioravanti acknowledges support from the Nederlandse Organisatie voor Wetenschappelijk Onderzoek Vici grant 639.023.811.}}
\date{}

\begin{document}
\maketitle
\begin{abstract}
We study a class of {\em aggregation rules} that could be applied to ethical AI decision-making. These rules yield the decisions to be made by automated systems based on the information of profiles of preferences over possible choices. We consider two different but very intuitive notions of preferences of an alternative over another one, namely {\it pairwise majority} and {\it position} dominance. Preferences are represented by permutation processes over alternatives and aggregation rules are applied to obtain results that are socially considered to be ethically correct. In this setting, we find many aggregation rules that satisfy desirable properties for an autonomous system. We also address the problem of the stability of the aggregation process, which is important when the information is variable. These results are a contribution for an AI designer that wants to justify the decisions made by an autonomous system. 
\end{abstract}

\section{Introduction}\
Many relevant scenarios arising from the use of Artificial Intelligent (AI) systems require making ethical decisions. So for instance, consider the case of an autonomous vehicle suffering a brake failure amid heavy urban traffic. Should it swerve its course and hit a wall, killing its passengers? Or stay its way and kill two pedestrians who are crossing the street? Of course, a definite answer depends on the availability of several pieces of information that are missing in this story, such as, for example, the age or occupation of both the passengers and pedestrians. But even with access to complete information about the situation, it is difficult to make a decision from an ethical point of view. This kind of ethical decision-making problem has long constituted a great challenge for Artificial Intelligence \citep{wallach2008moral}. Even if ground-truth ethical principles were available, the lack of corresponding formal specifications makes it difficult to solve these problems.\\

Some authors suggest that ``an approximation as agreed upon by society'' must be applied when ethical principles are not available \citep{dwork2012fairness}. \cite{conitzer2017moral} discuss a game theoretic and a machine learning approach to the development of a general framework, instead of appealing to the use of ad-hoc rules in each scenario. Furthermore, it would also be convenient to automatize the decision-making process, aggregating the diverse opinions held in society on ethical dilemmas \citep{greene2016embedding,conitzer2017moral}. In this sense, \cite{rossi2016moral} discuss related problems where preferences (not necessarily on moral issues) are aggregated by morally ranking all the alternatives.\\

\cite{noothigattu2018voting} propose and implement a concrete approach to ethical decision-making based on tools drawn from Computational Social Choice \citep{brandt2012computational,brandt2016handbook}. Specifically, they indicate that autonomous systems should be trained on a model of collective preferences so that when they face specific ethical dilemmas, they can efficiently make the desirable choices. The process of learning the model follows four steps (see \cite{noothigattu2018voting} for more details and discussions on every step and motivations):\\

\begin{itemize}
\item[{\it i.}] \textit{Data Collection}: human voters are asked to compare (a finite set of) pairs of alternatives,
\item[{\it ii.}] \textit{Learning}: a model of each individual preference is generated,
\item[{\it iii.}] \textit{Summarization}: a single model that approximates the preferences of all voters over all alternatives is created, and
\item[{\it iv.}] \textit{Aggregation}: a voting rule aggregates the individual preferences on a specific subset of alternatives (depending on the ethical dilemma faced) into a collective decision.
\end{itemize}

Notice that these steps involve the definition of an {\em aggregation rule}, i.e., a mathematical procedure to combine the preferences over decisions to be made over alternative decisions in ethical settings. We seek to define aggregation processes yielding solutions as close as possible to an `approximate' ground-truth ethical principle. {The ultimate goal is to provide a foundation for the automation of ethical decision-making in situations such as the failure of the braking system of an autonomous car \citep{piano2020ethical}, the automatic assignation of priorities in kidney exchange programs \citep{freedman2020adapting}, the assistance in judicial deliberation \citep{aletras2016predicting,berk2021fairness}, or more controversially, in the case of governments that use AI to improve their legitimacy \citep{starke2020artificial}.}\\


To address this general goal, we need to establish a model of preferences for each voter. A vast literature uses random utility models for the elicitation of preferences \citep{horowitz1994advances,parkes2012random,azari2013generalized}. That is, an agent's preference is modeled by drawing, for each alternative, a real value from a parameterized distribution. The values assigned by this distribution generate a ranking of the alternatives. In our setting, since an autonomous system may face different scenarios, each of which with a finite set of alternatives (that may belong to an infinite set), a more general notion is needed \citep{caron2012bayesian}. That is why we consider general models known as permutation processes \citep{noothigattu2018voting}. These models allow us to generate rankings over an infinite set of alternatives (for example, if age is considered as a continuum), even if only a finite subset is relevant (as in the case of a car that may hit only a small number of people).\\ 

Once we have a model for each voter, as well as a summary of those models (Step {\it iii}), we need a voting rule\footnote{More formally known as \textit{Social Choice correspondence}.} to select the best alternatives. We consider two different conceptions of what it means that an alternative is the `best' one, adapting to this setting the ideas of dominance among alternatives of \cite{caragiannis2016noisy}. We present some efficiency properties that a good voting rule should satisfy in terms of these dominance notions. While in  `classical' domains, only a few rules satisfy them, in the case of permutation processes we prove many efficient rules exist. Moreover, all of them select the same alternative.\\

We also consider the property of stability, which ensures the consistency of decisions across different scenarios of variable information. Again, in the specific setting of permutation processes, many rules verify this property.\\ 


This paper may contribute in different ways to the design of automated decision-making systems. First, we extend the theory of aggregation of permutation processes \citep{prasad2015distributional,marden2019analyzing}, that \cite{noothigattu2018voting} have shown to be adequate models of decision-making processes in autonomous systems. Besides the emphasis on ethical dilemmas, we also present a `technical' contribution to the discipline of \textit{Machine Ethics}, the field concerned with the `ethical behavior' of autonomous intelligent systems operating in social environments \citep{winfield2014towards,dennis2015towards,moor2020mature}. \cite{bjorgen2018cake} have proposed that the estimation of the ethical performance of an autonomous system should be done based on the responses to some ethical dilemmas that they consider as benchmarks. The results presented in this work may help to set up these `benchmarks'. Finally, this paper contributes to the goal of achieving AI value alignment, which is the problem of ensuring that AI's are properly aligned with human values \citep{russell2019human}. \cite{gabriel2020artificial} considers that voting theory is necessary to address this goal, but the impossibility results in this field are not very promising. Our results avoid many impossibilities and thus facilitate the work of an AI designer who wants to justify the decisions of an autonomous system. Roughly speaking, as many (well-known) voting rules are suitable for its use by an automated system, the AI designer has at hand all the properties verified by these rules, to use as reasons why a specific decision has been made.\\

The plan of the paper is as follows. In Section \ref{model}, we introduce the general model of the aggregation of permutation processes. Sections \ref{pmdom} and \ref{posdomsec} introduce different dominance definitions and important properties that voting rules should verify. In Section~\ref{aggrepermsec} we consider the setting in which the aggregation of permutations yields positive results. In Section~\ref{stabsec} we introduce the concept of stability and find some rules satisfying it. Finally, Section \ref{concsec} concludes with a brief discussion about some future lines of research. Most of the definitions of voting rules and all the proofs can be found in the Appendix.

\section{Model}\label{model}
Following the line of \cite{prasad2015distributional} and \cite{noothigattu2018voting}, we consider a potentially infinite set of alternatives $X$. Let $A\subseteq X$ be a finite subset of size $k$. A total ordering over $A$ can be characterized by a bijection $\sigma: A \rightarrow \{1, \ldots, k\}$ such that $\sigma(a)=j$ indicates that the position of alternative $a$ in $\sigma$ is $j$. By a slight abuse of language, we call $\sigma$ a {\em permutation} of $A$.\footnote{We will also refer to permutations over sets of alternatives as votes or rankings.} Let $S_A$ be the set of all permutations of $A$. For each $\sigma \in S_A$, $\sigma(a)<\sigma(b)$ is interpreted as indicating that $a$ is more preferred than $b$ in the order characterized by $\sigma$. We denote this as $a>_{\sigma}b$.\\ 

With $\sigma|_B$ we denote the restriction of $\sigma$ to $B\subseteq A$. Given a probability distribution $P$ over $S_A$ and $B\subseteq A$, we define
$$P_B(\sigma')=\sum_{\sigma\in S_{A}:\sigma|_B=\sigma'}P(\sigma)$$
where $\sigma'\in S_B$. \\

A {\em permutation process} $\Pi$ is a collection of probability distributions over sets of permutations $S_B$ such that $B \subseteq A$, with $\Pi(A)$ being a distribution over $S_A$. A permutation process is {\it consistent} if $\Pi(A)|_B=\Pi(B)$ for any finite subset $B\subseteq A$, that is if the distribution over $S_B$ is obtained by marginalizing out the extra alternatives in $A\setminus B$ from the distribution over $S_A$. \\




A particularly useful representation of consistent permutation processes is in terms of utilities \citep{marden2019analyzing}. That is, given a stochastic process $U$, indexed by $X$, such that for any \mbox{$A=\{x_1,x_2,\ldots,x_k\}\subseteq X$}, we can define the probability of $\sigma\in S_A$ as the probability that \mbox{$sort(U_{x_1},U_{x_{2}},\ldots,U_{x_k})$} coincides with that of $>_{\sigma}$, where $sort$ is an operation that yields a linear ordering of $\{U_{x_j}\}_{j=1}^n$. Thus, $U: A \rightarrow \mathbb{R} $ can be conceived as a stochastic utility function.\\

We follow the line of \cite{parkes2012random} and focus our attention on permutation processes in which the random utilities are independently drawn from distributions in the exponential family like the Normal, Poisson, Gamma, Binomial, and Negative Binomial distributions \citep{morris1982natural}.\\

Consider a finite set of $N$ voters. The preferences of a voter $i$ over a finite set of alternatives $A$, is denoted by $\sigma_i\in S_A$. A preference profile is defined as a collection of $N$ individual rankings, $\bar{\sigma} = (\sigma_1,\ldots,\sigma_N)$. In our setting, the identity of the voters is not important. So we can consider an anonymous profile $\pi_A\in[0,1]^{|A|!}$, where, for each $\sigma\in S_A$, $\pi_A(\sigma)\in[0,1]$ is the probability that a random voter has the ranking $\sigma$.\footnote {We omit subscripts unless they are specifically needed.} Without risk of confusion, with $\sigma\in\pi_A$ we refer to the orders $\sigma\in S_A$ such that $\pi_A(\sigma)>0$.\\

The relation between an anonymous preference profile and a permutation process over a finite subset $A\subseteq X$, is that for $\sigma\in S_A$, $\pi_A(\sigma)$ is the probability of $\sigma$ in $\Pi(A)$\footnote{Accordingly, we refer to $\Pi(A)$ as an anonymous preference profile.}.\\

We want to analyze aspects of the aggregation of preferences induced by a permutation process, defined by the selection of the winning alternatives. We define a \textit{social choice correspondence} (SCC) as a function $f$ that maps an anonymous preference profile defined over a finite subset $A\subseteq X$ into a nonempty subset of $A$. Examples of SCC are the \textit{plurality} rule, which selects the alternatives that are on the top of the largest number of individual preferences, $argmax_{a\in A}\sum_{\sigma\in S_A:\sigma(a)=1}\pi(\sigma)$, or the \textit{antiplurality} rule, that selects the alternatives that are the least preferred by the smallest number of individual preferences,  $argmin_{a\in A}\sum_{\sigma\in S_A:\sigma(a)=k}\pi(\sigma)$.\\

\section{Pairwise Majority Dominance}\label{pmdom}
We can define what it means for an alternative to be better than another one in aggregate terms, based on the concept of {\it majoritarian power}. That is, an alternative $a$ is better than an alternative $b$, if $a$ is preferred to $b$ by a majority of the individuals, i.e., a group of cardinality larger or equal than $|\frac{N}{2}|$. We have the following notion of dominance: 

\begin{defi}
An alternative $a\in A$ pairwise majority-dominates (PM-dominates) another alternative $b\in A$ in an anonymous preference profile $\pi$ over $A$, denoted  $a\rhd^{pm}_{\pi}b$, if $|\{\sigma\in \pi:a>_{\sigma}b\}|\geq|\{\sigma\in \pi:b>_{\sigma}a\}|$.
\end{defi}

This binary relation is complete but not transitive, due to the possible existence of Condorcet cycles.\footnote{A Condorcet cycle or `Condorcet Paradox' occurs when $a\rhd^{pm}_{\pi}b$, $b\rhd^{pm}_{\pi}c$ and $c\rhd^{pm}_{\pi}a$.}\\

The following is a desirable property to be satisfied by a social choice correspondence. Intuitively, it means that the voting rule is `consistent' with the dominance relation: if it selects an alternative $b$, then it must also select the alternatives that are better than $b$.

\begin{defi}
An anonymous SCC $f$ is said to be pairwise majority-dominance-efficient (PMD-efficient) if for every anonymous preference profile $\pi$ and any two alternatives $a,b\in A$, if $a\rhd^{pm}_{\pi}b$, then $b\in f(\pi)$ implies that $a\in f(\pi)$.
\end{defi}

A stronger requirement states, roughly, that dominated alternatives should not be selected, unless they also dominate the alternatives that dominate them.

\begin{defi}
An anonymous SCC $f$ is said to be strongly PMD-efficient if for every anonymous preference profile $\pi$ over $A$, and any two alternatives $a,b\in A$ such that $a\rhd^{pm}_{\pi}b$, we have:\\
-if $b\ntriangleright_{\pi}^{pm}a$, then $b\notin f(\pi)$,\\
-if $b\rhd_{\pi}^{pm}a$, then $b\in f(\pi)$ if and only if $a\in f(\pi)$.
\end{defi}

Strong PMD-efficiency implies PMD-efficiency. This notion sets the bar too high on SCCs since there are no rules satisfying it on unrestricted domains. This problem appears because Condorcet cycles can arise in those domains. But even efficiency is a very demanding condition:\footnote{The definitions of the voting rules can be found in the Appendix. The definitions may also be found in \cite{fishburn1977condorcet} and \cite{felsenthal2014weak}.}

\begin{propi}\label{pocasopciones}
The Black, Dodgson, Young, Kemeny, Nanson, Minimax, and Fishburn rules are not PMD-efficient voting rules.
\end{propi}

This means that this notion of efficiency is not satisfied by many (important and well-studied) voting rules. An alternative intuition is that an ``efficient'' rule should select the `most' dominating alternatives, or at least, the ones that are the `least' dominated. The {\em Schwartz set} captures this idea \citep{schwartz1972rationality}:

\begin{defi}
The \textit{Schwartz set} of a profile of preferences $\pi$, denoted $Sc(\pi)$, is the minimal set $Sc(\pi) \subseteq A$ such that $ Sc(\pi) \neq \emptyset$ verifying that for any $a \in Sc(\pi)$ there is no $b \in A \setminus Sc(\pi)$ such that $b\rhd_{\pi}^{pm}a$.
\end{defi} 

This set is nonempty for any profile of preferences. The \textit{Schwartz voting rule} is defined as the SCC that selects as winning alternatives the entire Schwartz set. The following result follows immediately from the definition of a Schwartz set:

\begin{propi}
The Schwartz voting rule is PMD-efficient.
\end{propi}

\section{Position Dominance}\label{posdomsec}
In this section, we consider a different notion of dominance. Namely, we say that an alternative $a$ dominates an alternative $b$ if $a$ is positioned higher in more rankings than $b$. Formally:\footnote{We adapt the definition introduced by \cite{caragiannis2016noisy}, allowing ties.}

\begin{defi}
Given an anonymous preference profile $\pi$ on $A$, an alternative $a\in A$, $i\in N$ and $j\in\{1,\ldots,|A|\}$, let $ s_j(a)=|\{i:\sigma_i(a)\leq j\}|$. That is, $s_j(a)$ is the number of voters that rank alternative $a$ at a position $j$ or lower. For $a,b\in A$, we say that $a$ position-dominates (Pos-dominates) $b$, denoted by $a\rhd^{pos}_{\pi}b$, if $s_{j}(a)\geq s_{j}(b)$ for every $j\in\{1,\ldots,|A|\}$.

\end{defi}
This binary relation is not complete, since there may exist alternatives that are not comparable under a given profile. But Pos-dominance is transitive, ruling out the possibility of cycles.\\

The following example shows that neither one of the dominance relations presented above (PM and Pos) implies the other.

\begin{exa}
In the preference profile $\{a>_1 b >_1 c,\ a>_2 b >_2 c,\ b>_3 c>_3 a\}$ we have that $a\rhd^{pm}_{\pi}b$ and $a\ntriangleright_{\pi}^{pos}b$, while in the profile $\{a>_1 c>_1 d>_1 b,\ b>_2 a >_2 c >_2 d,\ c>_3 d>_3 b>_3 a\}$ we have $a\rhd^{pos}_{\pi}b$ and $a\ntriangleright_{\pi}^{pm}b$.
\end{exa}

The properties of Position Dominance-efficient (PosD-efficient) SCC and strongly Position Dominance-efficient (strongly PosD-efficient) SCC are analogous to those of PMD-efficient and strongly PMD-efficient SCC. We have that:\footnote{The definitions of the voting rules in this section can be found in the Appendix.} 

\begin{propi}\label{buck}
The Bucklin SCC is PosD-efficient but not strongly PosD-efficient.
\end{propi}

Moreover, there exists an entire family of PosD-efficient rules:

\begin{propi}\label{effiposd}
All positional scoring rules are PosD-efficient.
\end{propi}



In turn, the strong PosD-efficiency of scoring rules depends on their associated weights:

\begin{propi}\label{strongeffiposd}
Scoring rules with associated decreasing weight vectors are strongly PosD-efficient.
\end{propi}

Notice that \cite{noothigattu2018voting} introduced a stronger notion,  {\it swap-dominance}, according to which an alternative $a$ dominates alternative $b$ if every ranking that places $a$ above $b$ has at least as much weight as the ranking obtained by swapping the positions of $a$ and $b$, and keeping everything else fixed. That is, for every $\sigma\in S_A$ such that $a>_\sigma b$, we have that $\pi(\sigma)>\pi(\sigma')$, where $\sigma'$ is the ranking $\sigma$ with $a$ and $b$ swapped. Swap dominance implies both pairwise majority- and position-dominance. Most of our results are still valid if we consider this notion.

\section{Aggregation of Permutation Processes}\label{aggrepermsec}
So far, we introduced two different notions of dominance, rather limited in their efficiency. For richer results, we can turn our attention to the {\em aggregation} of permutation processes. We will show that if the permutation process verifies a very natural property, many important rules become efficient and they all select the same best alternatives. 

\begin{defi}
An alternative $a\in X$ PM-dominates (respectively Pos-dominates) an alternative $b\in X$ in the permutation process $\Pi$, denoted  $a\rhd_{\Pi}^{pm} b$ $(a\rhd^{pos}_{\Pi}b)$, if for every finite set of alternatives $A\subseteq X$ such that $a,b\in A$, we have that $a$ PM-dominates (Pos-dominates) $b$ in the anonymous preference profile $\Pi(A)$.
\end{defi}

\begin{defi}
A permutation process $\Pi$ over $X$ is said to be PM-compatible (Pos-compatible) if for every $A\subseteq X$, the binary relation $\rhd_{\Pi}^{pm}|_{A} $ $(\rhd^{pos}_{\Pi}|_{A})$ is a total preorder over $A$.\footnote{A total preorder is a transitive and complete binary relation.}
\end{defi}

When a permutation process is PM-compatible, the total preorder is consistent across all the different finite subsets of $X$:

\begin{lema}\label{consist}
	Let $\Pi$ be a consistent permutation process that is PM-compatible. Then, for any finite subset of alternatives $A\subseteq X$, $(\rhd_{\Pi(A)}^{pm})=(\rhd_{\Pi}^{pm}|_{A})$
\end{lema}

This result allows us to write $\Pi$ instead of $\Pi(A)$ whenever we are referring to PM-dominance relations. The same is not true for a permutation process that is Pos-compatible. The dominance order of the alternatives may not be consistent between different subsets of $X$, as shown in the following example.

\begin{exa}
Consider the set of alternatives $A=\{a,b,c\}$ and the profile of preferences where one voter chooses $a>b>c$ and another one $b>c>a$. Then we have that the only possible order is $b\rhd_{\Pi(A)}^{pos}a\rhd_{\Pi(A)}^{pos}c$. If we now consider $B=\{a,b\}$, we obtain that  $b\rhd_{\Pi(B)}^{pos}a$ and $a\rhd_{\Pi(B)}^{pos}b$.
\end{exa}

We can introduce a new definition of compatibility, according to which the orders are consistent over all the subsets.
\begin{defi}
A permutation process $\Pi$ over $X$ is said to be strongly PM-compatible (Pos-compatible) if it is PM-compatible (Pos-compatible) and the preorder $\rhd_{\Pi}$ is consistent across all subsets  $A\subseteq X$.	
\end{defi}

It is clear from Lemma \ref{consist} that if a permutation process is PM-compatible it is also strongly PM-compatible.\\

The next result states that if the permutation process is compatible, then under an efficient voting rule one of the winning alternatives is the one that dominates the rest of the alternatives.

\begin{teo}\label{altgan}
	Let $f$ be an anonymous PMD-efficient SCC, and let $\Pi$ be a PM-compatible permutation process. Then, for any finite subset of alternatives $A$, there exists an $a\in A$ such that $a\rhd_{\Pi}^{pm}b$ for all $b\in A$. Moreover, $a\in f(\Pi(A))$.\\
	
	Let $f$ be an anonymous PosD-efficient SCC, and let $\Pi$ be a Pos-compatible permutation process. Then, for any finite subset of alternatives $A$, there exists an $a\in A$ such that $a\rhd_{\Pi(A)}^{pos}b$ for all $b\in A$. Moreover, $a\in f(\Pi(A))$.\\
	
	If $\Pi$ is strongly Pos-compatible, there exists an $a\in A$ such that $a\rhd_{\Pi}^{pos}b$ for all $b\in A$. Moreover, $a\in f(\Pi(A))$
\end{teo}

In the light of Proposition \ref{pocasopciones}, it may seem that Theorem \ref{altgan} is not very relevant since there are not too many PM-efficient voting rules. But this is no longer the case, as shown by the following result, for PM-compatible permutation processes. As a first step, let us recall that, according to \cite{fishburn1977condorcet}, a rule verifies the Strict Condorcet Principle if it selects the alternatives that beat or tie with every other candidate.\footnote{The undefeated alternatives are called in the literature as Weak Condorcet winners or Quasi-Condorcet candidates.}

	Let $\Pi$ be a PM-compatible and consistent permutation process. Then the Black, Nanson, Dodgson, Young, Minimax, Kemeny, Fishburn and Schwartz rules select $f(\Pi(A))=\{a\in A:a\rhd_{\Pi}^{pm}b\ \mbox{for all} \ b\in A\}$. Moreover, these rules are strongly PM-efficient.
\begin{teo}\label{strongeffi2}
	Let $\Pi$ be a PM-compatible and consistent permutation process. Then a rule is strongly PM-efficient if and only if it verifies the Strict Condorcet Principle.
\end{teo}

Note that according to the results in \cite{fishburn1977condorcet} and \cite{felsenthal2014weak}, the Black, Dodgson, Young, Minimax, and Fishburn rules verify the Strict Condorcet Principle. Under the absence of cycles, the Nanson, Kemeny, and Schwartz rules also verify the Strict Condorcet Principle.
\\

For position-dominance, we obtain an analogous result, but under the proviso that the order of alternatives may not be consistent among subsets of $X$. 

\begin{teo}\label{poscom}
Let $\Pi$ be a Pos-compatible and consistent permutation process, $f$ a scoring rule with an associated decreasing weight vector, and $A$ a finite set of alternatives. Then  $f(\Pi(A))=\{a\in A:a\rhd_{\Pi(A)}^{pos}b \ \mbox{for all} \ b\in A\}$.
\end{teo}

As a consequence of Theorem \ref{poscom}, we get a positive result with strong compatibility, since the orders are consistent across different subsets containing the same alternatives.

\begin{coro}\label{coropos}
Let $\Pi$ be a strongly Pos-compatible and consistent permutation process and $f$ a scoring rule with an associated decreasing weight vector. Then for any finite set of alternatives $A\subseteq X$, $f(\Pi(A))=\{a\in A:a\rhd_{\Pi}^{pos}b \ \mbox{for all} \ b\in A\}$.
\end{coro}

Theorems \ref{strongeffi2}, \ref{poscom} and Corollary \ref{coropos} are particular cases of a stronger result:

\begin{teo}\label{general}
Let $\Pi$ be a strongly compatible and consistent permutation process and $f$ a strongly efficient SCC. Then for any finite set of alternatives $A\subseteq X$, $f(\Pi(A))=\{a\in A:a\rhd_{\Pi}b \ \mbox{for all} \ b\in A\}$.\footnote{Whenever the type of dominance relation is not explicitly specified, the result must be valid for both PM and Pos dominance.}
\end{teo}

Then, if the permutation process is compatible, an efficient social choice correspondence selects the alternatives that beat or tie with every other alternative. The following example, due to \cite{caragiannis2016noisy} shows that even when the permutation process is compatible, the `best' alternative may not coincide under the two dominance relations.

\begin{exa}\label{ejpmposd}
Consider $\Pi$, the consistent permutation process which given the alternatives $a,b$ and $c$, yields the following profile:  $(a>_1 b>_1 c)$, with weight $\frac{4}{11}$;  $(b>_2 a>_2 c)$, with weight $\frac{2}{11}$; $(b>_3 c>_3 a)$ with weight $ \frac{3}{11} $ and $(c>_4 a>_4 b)$ weighted $\frac{2}{11}$. Then we have that  $a\rhd_{\Pi}^{pm}b\rhd_{\Pi}^{pm}c$ and $b\rhd_{\Pi}^{pos}a\rhd_{\Pi}^{pos}c$.\footnote{It is easy to show that a `swap-dominance' ranking cannot be obtained. For this to happen, a necessary condition would be the coincidence between the `majority' and `position' rankings. The following example shows that it is not sufficient: in the anonymous profile \mbox{$\{a>_1b>_1c,a>_2c>_2b,c>_3b>_3a$\}} the `majority' and `position' rankings coincide, but still no swap-dominance ranking obtains.}
\end{exa}

While the two dominance relations may not yield the same outcome, the question of when a permutation process is compatible is relevant. To answer it recall that a permutation process can be interpreted in terms of utilities, which allows us to introduce the following definition.
\begin{defi}
Alternative $a\in X$ dominates $b\in X$ in the utility process $U$ if for every finite subset of alternatives containing $x_1 = a$ and $x_2 = b$,  $\{x_1, x_2,\ldots,x_{m}\}\subseteq X$, and every vector of utilities $(u_1,u_2,\ldots,u_m)$ with $u_1\geq u_2$ we have that
	$$p(u_1,u_2,\ldots,u_m)\geq p(u_2,u_1,\ldots,u_m) $$
where $p$ is the probability function.
\end{defi}

The following two lemmas state that `natural' utility processes are indeed compatible.

\begin{lema}\label{consistencia}
Let $\Pi$ be a consistent permutation process and $U$ be its corresponding utility process. If alternative $a$ dominates $b$ in $U$, then $a\rhd_{\Pi}b$.
\end{lema}

We consider consistent permutation processes $U$ such that given a set of alternatives $\{x_1,\ldots,x_m\}$, the utilities $(U_{x_1},\ldots,U_{x_m})$ are independent and have a distribution drawn from the Exponential Family \citep{morris1982natural}:
$$p_{U_{a}}(u_{1})=\exp({\eta(\mu_{a})T(u_{1})-A(\mu_{a})+B(u_{1})})$$
Examples of this are the Gaussian or Normal (Thurstone-Mosteller Process \citep{thurstone1927law,mosteller2006remarks}), the Gumbel (Placket-Luce Process \citep{plackett1975analysis,luce2012individual}), Poisson, Gamma, Binomial and Negative Binomial distributions. Then:

\begin{lema}\label{efresults}
Under a utility process  with a distribution belonging to the Exponential Family, the alternative \textit{a} dominates alternative \textit{b} if $\eta(\mu_{a})>\eta(\mu_{b})$ and $T(u_{1})>T(u_{2})$.\footnote{For the  Normal $(N(\mu_{x_i},\frac{1}{2}))$ and the Gumbel $(G(\mu_{x_i},\gamma))$ processes this means that $\mu_a\geq \mu_b$. For a Poisson process $(P(\lambda_{x_i}))$, it means that $\lambda_a\geq \lambda_b$. For a Gamma process with fixed shape $(\Gamma(r,\lambda_{x_i}))$, the implication is that $\lambda_b\geq \lambda_a$. For a Binomial process $(B(n_{x_i},p_{x_i}))$ with fixed $n$, it means that $p_a\geq p_b$. If $p$ is fixed and $n$ is variable, it corresponds to $n_a\geq n_b$. For a Negative Binomial process $(NB(r_{x_i},p_{x_i}))$ with fixed $r$, it means that $p_b\geq p_a$ while if $p$ is fixed and $r$ is variable, it is $r_a\geq r_b$.}
\end{lema}

Lemma \ref{efresults} implies that a consistent permutation process with a utility process drawn from the Exponential Family (which includes some of the best-known distributions), is compatible under both notions of dominance. This ensues from the fact that  dominance only depends on the parameters of the distributions, all of which have real values (and thus are elements of a complete preorder). Moreover, Lemma \ref{consistencia} states that the preorders are the same under both dominance notions satisfying strong compatibility. This allows us to get rid of a fixed $A$ even for Pos-dominance. The downside of this is that we may rule out permutation processes that could be compatible under one dominance notion but not under the other.\\

According to Theorem \ref{strongeffi2}, Corollary \ref{coropos}, and Lemma \ref{efresults}, in many cases of interest we only need to find the most dominant alternatives, which, depending on the permutation process, are the alternatives with a maximum or minimum defined parameter.

\section{Stability}\label{stabsec}
In this section, we introduce another property that we expect a SCC should verify, namely the consistency of the voting rule across decisions.

\begin{defi}
Let $f$ be an anonymous SCC and $\Pi$ a permutation process over $X$. We say that $f$ is stable if for any non-empty and finite subset of alternatives $A$ and $B$ such that $B\subseteq A\subseteq X$, $f(\Pi(A))\cap B=f(\Pi(B))$ whenever $f(\Pi(A))\cap B\neq \emptyset$.
\end{defi}

Stability is a highly relevant property since an intelligent agent should be able to make decisions, even in the context of missing or noisy information \citep{greene2016embedding}. In our setting, this means that voting rules must yield consistent choices on restricted sets of alternatives. That is, even if `small' differences exist among the sets of alternatives the decisions made will not exhibit `big' differences. For example, the decision of an autonomous vehicle carrying three passengers on whether to go straight or swerve when no pedestrians are in front of it should be the same as if it were transporting only two people. \\

The next result shows that the stability of a SCC $f$ depends on how $f$ and the permutation process are defined.


\begin{teo}\label{stability}
Let $f$ be a strongly efficient SCC and let $\Pi$ be a strongly compatible permutation process. Then $f$ is stable.
\end{teo}

This is a very encouraging result, since according to Proposition \ref{strongeffiposd} and Theorem \ref{strongeffi2} there exist many well-known voting rules that are strongly efficient. Thus, these rules must also verify stability.

\section{Conclusions}\label{concsec}
In this paper, we extended the study of the method introduced by \cite{noothigattu2018voting} for automating ethical AI decision-making under different notions of dominance. According to one of those notions, an alternative is better than another if a majority prefers it. According to the other notion, an alternative is dominant if it is better positioned than the other alternatives in most rankings. At first sight, the dominance relations here introduced, seem very intuitive and natural. The fact that the same alternative is chosen under ours and \cite{noothigattu2018voting} notions of dominance justifies the choice of that alternative. We see, thus, this work as a contribution to the solution of the AI alignment problem in ethical decision-making.\\

We showed that, depending on how we learn the preferences of the voters, there are many well-known voting rules that behave well, in the sense that they select the most dominating alternatives. Moreover, it is only necessary to find the alternative with an optimum defined parameter (depending on the distribution used to learn the model). When we use `natural' distributions for the permutation process we do not need to distinguish between dominance notions. Another important property is stability, which amounts to the consistency of choices across different sets of alternatives. For example, an autonomous vehicle is likely to face potentially infinite scenarios, and we want decisions to be similar in similar scenarios. We found again that there are many voting rules verifying this property.\\


We consider this paper also as a contribution to the theory of aggregation of permutations. Its full development in this light may help to overcome classical problems of aggregation theory.\\

A possible future line of research consists in analyzing the strategic aspects of the aggregation of permutation processes. It is important to guarantee the non-manipulability of this decision-making process. For example, there should not exist advantages from lying about preferences (this property is known in the literature as strategy-proofness \citep{gibbard1973manipulation,satterthwaite1975strategy}) or from voting several times (known in the literature as false name-proofness \citep{conitzer2008anonymity,fioravanti2022false}). The results on manipulation avoidance are rather negative since there exist very few rules satisfying non-manipulability when considering unrestricted domains. We believe that in the case of the aggregation of permutation processes, new rules can be found to overcome these problems.
\bibliography{ref}

\appendix
\section{Voting rules}

\textbf{Black's rule:\footnote{The Black and Dodgson rules described are called Revised Black and Simplified Dodgson rules in \cite{felsenthal2014weak}.}} this rule is applied in two stages. In the first stage, the candidates that beat or tie with every other candidate are selected. If there are no such winners, each candidate receives a score depending on the positions on every ranking (the better the position the higher the score). The alternatives with the largest overall scores are elected.\\

\textbf{Dodgson's rule:} A Condorcet winner is a candidate that beats every other alternative. The rule computes the minimum number of times that it is necessary to swap two adjacent alternatives on some rankings to make each candidate a Condorcet winner. The alternatives that require the minimum steps are considered the winners.\\

\textbf{Young's rule:} This rule deletes candidates to make an alternative to the Condorcet winner. The candidates that require the least deletions are declared the winners.\\

\textbf{Kemeny's rule:} This rule selects the most preferred alternatives of the rankings that minimize the number of pairs of candidates that are ranked opposite by all the voters.\\

\textbf{Nanson's rule:} As in Black's rule, a score is given to the candidates according to their position. Every candidate with a score below the average score is deleted. The process is repeated with the candidates left, and so on. The rule selects the undeleted candidates.\\

\textbf{Minimax rule:} The candidates whose maximum losses in the paired comparisons are the least are declared winners.\\

\textbf{Fishburn's rule:} this rule checks out whether everything that beats \textit{x} also beats \textit{y} under simple majority, and if
\textit{x} beats or ties something that beats \textit{y}, then \textit{x} is ‘better than’ \textit{y} under simple majority. The rule selects the best alternatives under this notion.\\

\textbf{Bucklin's rule:} the Bucklin score of an alternative $a$, $B(a)$, is the minimum $k$ such that $a$ is among the first $k$ positions in the majority of input votes. The Bucklin rule selects the alternatives with the least Bucklin scores.\\

\textbf{Scoring rules:} these rules have an associated weights vector $ (\alpha_1, \alpha_2, \ldots, \alpha_{|A|})$, where $\alpha_i\geq \alpha_{i+1}$ for $i=1, \ldots, |A|-1$. An alternative $a$ in a profile $\sigma_j$ gets $\alpha_i$ points if $\sigma_j(a)=i$. The score of an alternative is the sum of all the points across all voters. The candidates with the most overall points are declared the winners. The Plurality and the Antiplurality rules are examples of scoring rules with associated weights vector $(1, 0, \ldots, 0)$ and $(1, \ldots, 1, 0)$ respectively. We say that a scoring rule has an associated decreasing weights vector if $\alpha_i> \alpha_{i+1}$ for $i=1, \ldots, |A|-1$.
\section{Proofs}
\begin{proof}Proposition \ref{pocasopciones}\\

Consider the profile $\sigma=(a>_1 b>_1 c>_1 d>_1 e, e>_2 d>_2 a>_2 c>_2 b, b>_3 c>_3 d>_3 e>_3 a)$. We have that $a\rhd^{pm}_{\pi}b$.\\

\textbf{Black's rule:} there is no alternative that beats or ties every other alternative. Then we use a scoring rule with an associated decreasing weighted vector, for example, $(5,4,3,2,1)$. We obtain that $score(a)=9$ and $score(b)=10$, with $b$ getting the largest score. Thus we have that $b\in Black(\pi)$ and $a\notin Black(\pi)$.\\

\textbf{Dodgson's rule:} alternative $b$ needs $1$ swap in the first vote to become the Condorcet winner, while all the other alternatives require $2$ or more. Then $b\in Dodgson(\pi)$ and $a\notin Dodgson(\pi)$.\\

\textbf{Young's rule:} alternative $b$ needs alternative $a$ to be eliminated in order to become the Condorcet winner, while all the other alternatives need $2$ or more alternatives to be eliminated. Then $b\in Young(\pi)$ and $a\notin Young(\pi)$.\\

\textbf{Kemeny's rule:} the ranking that minimizes the number of candidates that are ranked opposite by all the voters is $b>c>d>a$. Then $b\in Kemeny(\pi)$ and $a\notin Kemeny(\pi)$.\\

Consider the profile $\sigma=(e>_1 a>_1 b>_1 c>_1 d, b>_2 c>_2 d>_2 e>_2 a)$. We have that $a\rhd^{pm}_{\pi}b$.\\

\textbf{Nanson's rule:} we use a scoring rule with an associated decreasing weighted vector, for example, $(5,4,3,2,1)$. During the first round, $a$ and $d$ are eliminated. In the second round, $c$ is eliminated. Finally, we have that $b$ and $e$ are selected. Then $b\in Nanson(\pi)$ and $a\notin Nanson(\pi)$.\\

Consider the profile $\sigma=(d>_1 c>_1 a>_1 b, b>_2 c>_2 d>_2 a, d>_3 c>_3 a>_3 b, a>_4 b>_4 c>_4 d, b>_5 c>_5 d>_5a)$. We have that $a\rhd^{pm}_{\pi}b$.\\

\textbf{Minimax rule:} alternative $b$, $c$ and $d$ have $2$ as a maximum loss while $a$ has $3$. Then $b\in Minimax(\pi)$ and $a\notin Minimax(\pi)$.\\

Consider the profile $\sigma=(a>_1 b>_1 c>_1 d, b>_2 c>_2 d>_2 a)$. We have that $a\rhd^{pm}_{\pi}b$.\\

\textbf{Fishburn's rule:} in this profile $b$ beats $c$ and $d$, ties with $a$ and beat all the alternative that beat or tie with $a$. Then $b\in Fishburn(\pi)$ and $a\notin Fishburn(\pi)$.
\end{proof}
\begin{proof}Proposition \ref{buck}\\

Let $a\rhd^{pd}_{\pi}b $ and $b\in Bucklin(\pi)$. Let $B(b)=t$, that is, $s_{t}(b)=\alpha\geq[\frac{n+1}{2}]$ Because of position dominance, we have that $s_{t}(a)\geq s_{t}(b)=\alpha$. Then, $B(a)\leq t$. So $a\in Bucklin(\pi)$.\\

The following example shows that it is not strongly PMD-efficient. Let $\pi$ be the preference profile such that $\sigma=(a>_1 b>_1 c,a>_2 b>_2 c,b>_3 a>_3 c,c>_4 a>_4 b)$. Then $B(a)=B(b)=2$, $B(c)=3$, $s_{1}(a)=2$, $s_{1}(b)=$, $s_{2}(a)=4$ and $s_{2}(b)=3$. So, we have that $a\rhd^{pd}_{\pi}b$ and $b\ntriangleright^{pd}_{\pi}a$ but $Bucklin(\pi)=\{a,b\}$.
\end{proof}
\begin{proof}Proposition \ref{effiposd}\\

	For ease of demonstration we assume, without loss of generality, that $A=\{a,b,c\}$. The score of alternative $a$ is given by $Score(a)=(s_{1}(a),s_{2}(a)-s_{1}(b),3-s_{2}(a))\cdot(\alpha_{1},\alpha_{2},\alpha_{3})=(\alpha_{1}-\alpha_{2})s_{1}(a)+(\alpha_{2}-\alpha_{3})s_{2}(a)+3\alpha_{3}$ with $\alpha_{i}-\alpha_{i+1}\geq0$ for $i=1,2$. From the fact that $a\rhd^{pd}_{\pi}b $ we have that $s_{j}(a)\geq s_{j}(b)$ for $j=1,2$. Then $Score(a)\geq Score(b)$. If $b$ is such that $b\in f(\pi)$, then $a\in f(\pi)$ as it has more or equal points than $b$.
\end{proof}

\begin{proof} Proposition \ref{strongeffiposd}\\

	We assume again that $A=\{a,b,c\}$. And consider the same argument as in the proof of Proposition \ref{effiposd}. The main difference is that $\alpha_{i}-\alpha_{i+1}>0$ for $i=1,2$. Then if $a\rhd^{pd}_{\pi}b $ and $b\ntriangleright_{\pi}^{pd}a$, we have that $Score(a)> Score(b)$ and then $b\notin f(\pi)$. If instead $b\rhd^{pd}_{\pi}a$, then $Score(a)= Score(b)$ and both alternatives are selected if they have the most overall points.
\end{proof}
\begin{proof}Lemma \ref{consist}\\
First, by definition, we have that $a\rhd_{\Pi}^{pm}b$ implies $a\rhd_{\Pi(A)}^{pm}b$ for every subset $A\subseteq X$ such that $a,b\in A$. Then $(\rhd_{\Pi}|_{A})\subseteq (\rhd_{\Pi(A)})$.\\
	
	For the other inclusion, we must prove that $a\rhd_{\Pi(A)}^{pm}b$ implies that $a\rhd_{\Pi(C)}^{pm}b$ for any finite subset $C$ that contains $a$ and $b$. Let $a\rhd_{\Pi(A)}^{pm}b$ and suppose that $a\ntriangleright_{\Pi(C)}^{pm}b$ for $C$. Since $\rhd_{\Pi}^{pm}$ is PM-compatible, it is a total preorder for every $C$. Then it must be that $b\rhd_{\Pi(C)}^{pm}a$. Then
	\begin{equation}\label{igualdad}
	|\{\sigma\in \Pi(A):a>_{\sigma}b\}|\geq|\{\sigma\in \Pi(A):b>_{\sigma}a\}|
	\end{equation}
	and
	\begin{equation}\label{desigualdad}
		|\{\sigma\in \Pi(C):a>_{\sigma}b\}|<|\{\sigma\in \Pi(C):b>_{\sigma}a\}|
	\end{equation}
	Since the permutation process is consistent, when we restrict the domain to $\{a,b\}$, the relative positions of $a$ and $b$ do not change. So, (\ref{igualdad}) leads to $|\{\sigma\in \Pi({\{a,b\})}:a>_{\sigma}b\}|\geq|\{\sigma\in \Pi({\{a,b\})}:b>_{\sigma}a\}|$ and (\ref{desigualdad}) leads to \mbox{${|\{\sigma\in \Pi({\{a,b\})}:a>_{\sigma}b\}|<|\{\sigma\in \Pi({\{a,b\})}:b>_{\sigma}a\}|}$}. This is a contradiction, derived from the assumption that $a\ntriangleright_{\Pi(C)}^{pm}b$ for a finite subset $C$ containing $a$ and $b$. Then $a\rhd_{\Pi(C)}^{pm}b$ for any $C$, and thus  $a\rhd_{\Pi}^{pm}b$ verifies that $(\rhd_{\Pi(A)}) \subseteq (\rhd_{\Pi}|_{A})$.
\end{proof}
\begin{proof} Theorem \ref{altgan}\\
We prove the theorem for pairwise majority dominance since the proof is analogous to position dominance.\\

Since $\Pi$ is PM-compatible, the relation $\rhd_{\Pi}^{pm}$ restricted to $A$ is a total preorder. Therefore, there exists an alternative $a\in A$ such that $a\rhd_{\Pi(A)}^{pm}b$ for all $b\in A$.\\

The definition of an SCC states that the outcome is non-empty. Assume that there is a $b\neq a$ such that $b\in f(\Pi(A))$. Since $a\rhd_{\Pi(A)}^{pm}b$ and $f$ is PM-efficient, it must be that $a\in f(\Pi(A))$.  
\end{proof}
\begin{proof} Theorem \ref{strongeffi2}\\
Since $\Pi$ is PM-compatible, $\rhd_{\Pi(A)}$ is a total preorder, thus there are no `cycles'. By Theorem \ref{altgan}, there is at least one alternative $a\in A$ that PM-dominates the other alternatives. Moreover, if there is more than one, then these alternatives tie between them. As the outcome of a rule can not be the empty set, a strongly PM-efficient rule selects the alternatives that are tied between them and beat every other alternative. Then the rule verifies the Strict Condorcet principle. The only if part is straightforward.
\end{proof}
\begin{proof} Theorem \ref{poscom}\\
According to Theorem \ref{altgan} there exists an alternative $a\in A$ such that $a\rhd_{\Pi(A)}^{pos}b$ for all $b\in A$ and $a\in f(\Pi(A))$. Since $\Pi$ is Pos-compatible, any other alternative $b\in A$ is comparable with $a$.  By Proposition \ref{strongeffiposd}, all the scoring rules with associated decreasing weights are strongly PD-efficient. If $b\rhd_{\Pi(A)}^{pos}a$, it must be that $b\in f(\Pi(A))$, and $b\rhd_{\Pi(A)}^{pos}c$ for all $c\in A$ (because of transitivity). If $b\ntriangleright_{\Pi(A)}^{pos}a$, then $b\notin f(\Pi(A))$.
\end{proof}
\begin{proof} Theorem \ref{general}\\
Let $A$ be a finite set of alternatives. By Theorem \ref{altgan} and the fact that strong efficiency implies efficiency, we have that $\{a\in A:a\rhd_{\Pi}b\ \mbox{for all} \ b\in A\}\subseteq f(\Pi(A)) $.\\

Now let $a\in f(\Pi(A))$ and assume that there exists $b\in A$ such that $a\ntriangleright_{\Pi}b$. Since $\Pi$ is strongly compatible, $\Pi(A)$ has the same total preorder, for every $A\subseteq X$. Then $b\rhd_{\Pi(A)}a$ and $a\ntriangleright_{\Pi(A)}b$. Since $f$ is strongly efficient, it must be that $a\notin f(\Pi(A))$, a contradiction. Then $f(\Pi(A))\subseteq\{a\in A:a\rhd_{\Pi}b\ \mbox{for all} \ b\in A\}$.
\end{proof}
\begin{proof} Lemma \ref{consistencia}\\
Noothigattu et al \cite{noothigattu2018voting} prove in Lemma 4.9 that if $a$ dominates $b$ in $U$, then $a$ {\em swap-dominates} $b$. Swap dominance, as already noted, is a stronger notion that implies PM dominance and Pos-dominance. Then both conditions are satisfied.
\end{proof}
\begin{proof} Lemma \ref{efresults}\\
We have to find the conditions according to which $p(u_1,u_2,\ldots,u_m)\geq p(u_2,u_1,\ldots,u_m) $ when $u_1\geq u_2$. Since utilities are sampled independently, this implies checking when $p_{U_a}(u_1)p_{U_b}(u_2)\geq p_{U_a}(u_2)p_{U_b}(u_1)$. That is, 
$$\exp({\eta(\mu_{a})T(u_{1})-A(\mu_{a})+B(u_{1})})\exp({\eta(\mu_{b})T(u_{2})-A(\mu_{b})+B(u_{2}))}\geq$$
$$\exp({\eta(\mu_{a})T(u_{2})-A(\mu_{a})+B(u_{2})})\exp({\eta(\mu_{b})T(u_{1})-A(\mu_{b})+B(u_{1})})$$
\noindent and thus
$$\eta(\mu_{a})T(u_{1})-A(\mu_{a})+B(u_{1}))+\eta(\mu_{b})T(u_{2})-A(\mu_{b})+B(u_{2})\geq$$
$$\eta(\mu_{a})T(u_{2})-A(\mu_{a})+B(u_{2})+\eta(\mu_{b})T(u_{1})-A(\mu_{b})+B(u_{1})$$
Then
$$(\eta(\mu_{a})-\eta(\mu_{b}))T(u_{1})+(\eta(\mu_{b})-\eta(\mu_{a}))T(u_{2})\geq0$$
and
$$(\eta(\mu_{a})-\eta(\mu_{b}))(T(u_{1})-T(u_{2}))\geq0$$
Since $u_1\geq u_2$, it follows that $T(u_{1})\geq T(u_{2})$ and thus $\eta(\mu_{a})-\eta(\mu_{b})\geq0$ and $\eta(\mu_{a})\geq\eta(\mu_{b})$.\\

The proof for each particular probability distribution is analogous to the proof of the general case.
\end{proof}
\begin{proof} Theorem \ref{stability}\\
Let $B\subseteq A\subseteq X$ and assume that $f(\Pi(A))\cap B\neq \emptyset$. Let $a\in f(\Pi(A))\cap B$. By Theorem \ref{general} we have that $a\rhd_{\Pi}b$ for all $b\in A$, and thus, for all $b\in B$. Then $a\in f(\Pi(B))$.\\

Now consider $a\in f(\Pi(B))$. Then $a\in B$ and $a\rhd_{\Pi}b$ for all $b\in B$. Suppose that $a\notin f(\Pi(A))$. Then there is a $c\in A$ such that $a\ntriangleright_{\Pi}c$. Since $\Pi$ is strongly compatible, we have that $c\rhd_{\Pi(A)}a$. Since we assume that $f(\Pi(A))\cap B\neq \emptyset$, there is an alternative $d\in f(\Pi(A))\cap B$. Then $d\rhd_{\Pi(A)}c$ and by transitivity we have that $a\ntriangleright_{\Pi(A)}d$. Since $d\in B$ and $\Pi$ is strongly compatible, we have that $a\ntriangleright_{\Pi(B)}d$. But $f$ is strongly efficient, so we have that $a\notin f(\Pi(B))$, which is a contradiction.
\end{proof}
\end{document}